\documentclass{article}

\usepackage{amsmath}
\usepackage{amsthm}
\usepackage{amssymb}
\usepackage[noend]{algpseudocode}
\usepackage{algorithm}
\usepackage{float}
\usepackage{amssymb}
\usepackage{mathtools}
\usepackage{subcaption}
\newtheorem{theorem}{Theorem}
\newtheorem{lemma}{Lemma}
\newtheorem{definition}{Definition}
\newtheorem{fact}{Fact}
\newcommand{\R}{\mathbb{R}}
\providecommand{\norm}[1]{\left\lVert#1\right\rVert}

\usepackage{microtype}
\usepackage{graphicx}
\usepackage{booktabs}
\usepackage{hyperref}

\usepackage{fullpage}

\title{Sparse Coresets for SVD on Infinite Streams}

\author{

Vladimir Braverman
\thanks{Department of Computer Science, Johns Hopkins University.
}
\and
Dan Feldman\thanks{Department of Computer Science, University of Haifa.}
\and
Harry Lang\thanks{MIT CSAIL.}
\and
Daniela Rus\thanks{MIT CSAIL.}
\and
Adiel Statman\thanks{University of Haifa.}
}

\begin{document}

\maketitle

\begin{abstract}
In streaming Singular Value Decomposition (SVD), $d$-dimensional rows of a possibly infinite matrix arrive sequentially as points in $\mathbb{R}^d$.
An $\epsilon$-coreset is a (much smaller) matrix whose sum of square distances of the rows to any hyperplane approximates that of the original matrix to a $1 \pm \epsilon$ factor.
Our main result is that we can maintain a $\epsilon$-coreset while storing only $O(d \log^2 d / \epsilon^2)$ rows.
Known lower bounds of $\Omega(d / \epsilon^2)$ rows show that this is nearly optimal.
Moreover, each row of our coreset is a weighted subset of the input rows.
This is highly desirable since it: (1) preserves sparsity; (2) is easily interpretable; (3) avoids precision errors; (4) applies to problems with constraints on the input.
Previous streaming results for SVD that return a subset of the input required storing $\Omega(d \log^3 n / \epsilon^2)$ rows where $n$ is the number of rows seen so far.
Our algorithm, with storage independent of $n$, is the first result that uses finite memory on infinite streams.
We support our findings with experiments on the Wikipedia dataset benchmarked against state-of-the-art algorithms.
\end{abstract}


\newpage

\section{Introduction}

Coresets are a small representation of a large dataset that preserve some key property of interest.
For example, given a large clustering problem over $n$ points, one could build an $\epsilon$-coreset of $O(\log n)$ points and then run an off-the-shelf algorithm on the coreset to obtain a $(1+\epsilon)$-approximate solution.
Coresets have been widely studied in recent years for a variety of problems~\cite{FL11, BFL19, maalouf2019tight, FSS13}.

Much of the large scale high-dimensional
data sets available today (e.g. image streams, text streams, etc.) are sparse.
For a concrete example, we can associate a matrix with Wikipedia where ($1.4 \times 10^6$) words
define the columns and ($4.4 \times 10^6$) documents define the rows.
Entry $(i,j)$ is the number of occurrences of word $j$ in document $i$.
Since most documents only contain $\sim 10^3$ words, this matrix is very sparse.
Recently, the first algorithm that could compute the eigenvectors of this matrix on commodity hardware was provided by~\cite{FVR16}.
We benchmark our algorithm against this result for the Wikipedia matrix and show dramatic improvements in both runtime and accuracy.

We consider Singular Value Decomposition (SVD) in the streaming model, which is applicable to processing large datasets or real-time data~\cite{FVR16, maalouf2019tight}.
In the streaming model, the rows of a (possibly infinite) matrix $A \in \mathbb{R}^{? \times d}$ arrive sequentially as $(a_1, a_2, a_3, \ldots)$.
Our main contribution is an algorithm to maintain an $\epsilon$-coreset for SVD while storing only $O(d \log d)$ rows.
Our construction has the desirable property that it is a weighted subset of the input.
The property that each row of the coreset is a multiple of a single row of the input is important because: (1) sparsity is preserved, (2) the coreset is easy to interpret, and (3) there are less floating point precision issues.

\begin{theorem}[Main Theorem]
It is possible to maintain an $\epsilon$-coreset for SVD in the streaming model on an $n \times d$ matrix using $O(d \log d / \epsilon^2)$ space and $O(d^2)$ worst-case update time\footnote{The space and update time are measured in terms of a single $d$-dimensional row (i.e. multiply by $d$ obtain the space/time in bits)}.
\end{theorem}

Previous coreset results required dependence on $n$, the number of rows seen so far.
This is the first streaming result for an SVD coreset that uses finite memory on an infinite stream.
Moreover, the coreset is a weighted subset of the input, implying that properties of the input (such as sparsity) are preserved.

Our solution is a mixture of a known \texttt{RAM} model coreset with a novel streaming approach.
The existing construction defines a distribution over all rows of the input, and draws an i.i.d. sample which becomes the coreset (after reweighting).  Thus, the coreset is naturally a subset of the input.
We use the coreset inductively to compute the sampling probabilities after each update (i.e. after each row is received), and maintain a set of samplers in parallel.
Each sampler runs a simple procedure we call \textit{singleton sampling}.
If a sampler contains exactly one element, then the probability distribution regarding which element it contains is identical to the desired sampling distribution for the coreset.
If it does not contain exactly one element, then we ignore that sampler as having failed (at that particular time; later it may contain exactly one element and be used).
Our coreset is then the union over all samplers containing one element, therefore keeping the known \texttt{RAM} construction in the streaming setting without incurring any of the overhead associated with previous techniques.

\section{Prior Work}

\begin{center}
\begin{tabular}{ |c|c|c|c|c|c| } 
 \hline
 Work & Subset of Input & Space In Rows \\ 
 \hline
 \cite{FVR16} & Yes & $O(d^2 \epsilon^{-2} \log^3 n)$ \\ 
 \cite{CE14} & No & $O(d^2 \epsilon^{-2} \log^3 n)$ \\
 \cite{FSS13} & ? & $O(d \epsilon^{-2})$ \\
 \cite{FL11} & Yes & $O(d^2 \epsilon^{-2} \log^3 n)$ \\
 \cite{CP15} & Yes & $O(d \log d \epsilon^{-2} \log^3 n)$ \\
 \textbf{**} & Yes & $O(d \log d / \epsilon^2)$ \\ 
 \hline
\end{tabular}
\end{center}

The result most similar to the present work is~\cite{CP15}, where rows are sampled according to their \textit{Lewis weight}.
As explained in that paper, the Lewis weight is equivalent to the statistical leverage score defined in other literature such as~\cite{Woodruff2012}.
While sampling $O(d \log d \epsilon^{-2})$ rows, they construct a coreset in the \texttt{RAM} model.
Like all \texttt{RAM} model construction with inverse-quadratic dependence on $\epsilon$, the merge-and-reduce technique can be applied to construct a coreset in the streaming model while incurring $O(\log^3 n)$ space overhead.

Also in the \texttt{RAM} model, a coreset of size $O(d^2 /\epsilon^2)$ is given in~\cite{FVR16}.  When plugged into the merge-and-reduce tree, the final space is $O(d^2 \epsilon^{-2} \log^3 n)$ space.
In the current work, we provide a coreset using only $O(d \log d \epsilon^{-2})$ space, so independent of the input size $n$, therefore valid for infinite streams.

A dimensionality reduction technique for low-rank approximation~\cite{CE14} can be used to construct a coreset for SVD of size $O(d^2 / \epsilon^2)$ in the RAM model, leading to a streaming space of $O(d^2 \epsilon^{-2} \log^3 n)$.  However, the resulting coreset is not a subset of the input.  This property is attained (with the same space bound) in~\cite{FL11}.

\section{Our Techniques}

To construct a coreset, we must sample from a distribution over the input.
This is the approach taken by a wide-class of constructions such all those inspired by~\cite{FL11}.
In the \texttt{RAM} model this is trivial since we can access the entire input without restriction.
Suppose we wish to sample one point from a stream ($a_1$, $a_2$, \ldots, $a_n$) where the probability of sampling $a_i$ is $\pi(a_i) / \sum_{j = 1}^n \pi(a_j)$ for some non-negative function $\pi$.
If the distribution is known beforehand, then reservoir sampling will obtain the desired sample.
The difficulty in the streaming problem is that the distribution is discovered as the stream arrives.

Reservoir sampling does not work when the probabilities change.
We generalize reservoir sampling in exchange for small probability of failure.
See Figure~\ref{fig:sampling} for an example of our technique on a simple distribution over two points $(a, b)$.
Here $\pi(a) = 1/3$ and $\pi(b) = 2/3$.
Suppose that after receiving point $a$, we take a random draw $u_a$ from the interval $[0, 1)$ and keep $a$ in memory if and only if $u_a < \pi(a)$.  We do the same for $b$.
At the end of the stream, we declare failure unless there is exactly one point stored in memory (we will argue the probability of failure is low).
Then, as seen in the top row, the distribution that we actually sample from has been distorted (see top-right box).
However, if we pre-transform the distribution (as shown in the bottom row), then this cancels the distortion and we recover the desired distribution (see bottom-right box).

While this example does not show the probabilities $\pi(a)$ or $\pi(b)$ changing, we later apply some additional properties gauranteed by the functions relevant to SVD that we can easily handle this generalization.

\begin{figure}
  \includegraphics[width=\linewidth]{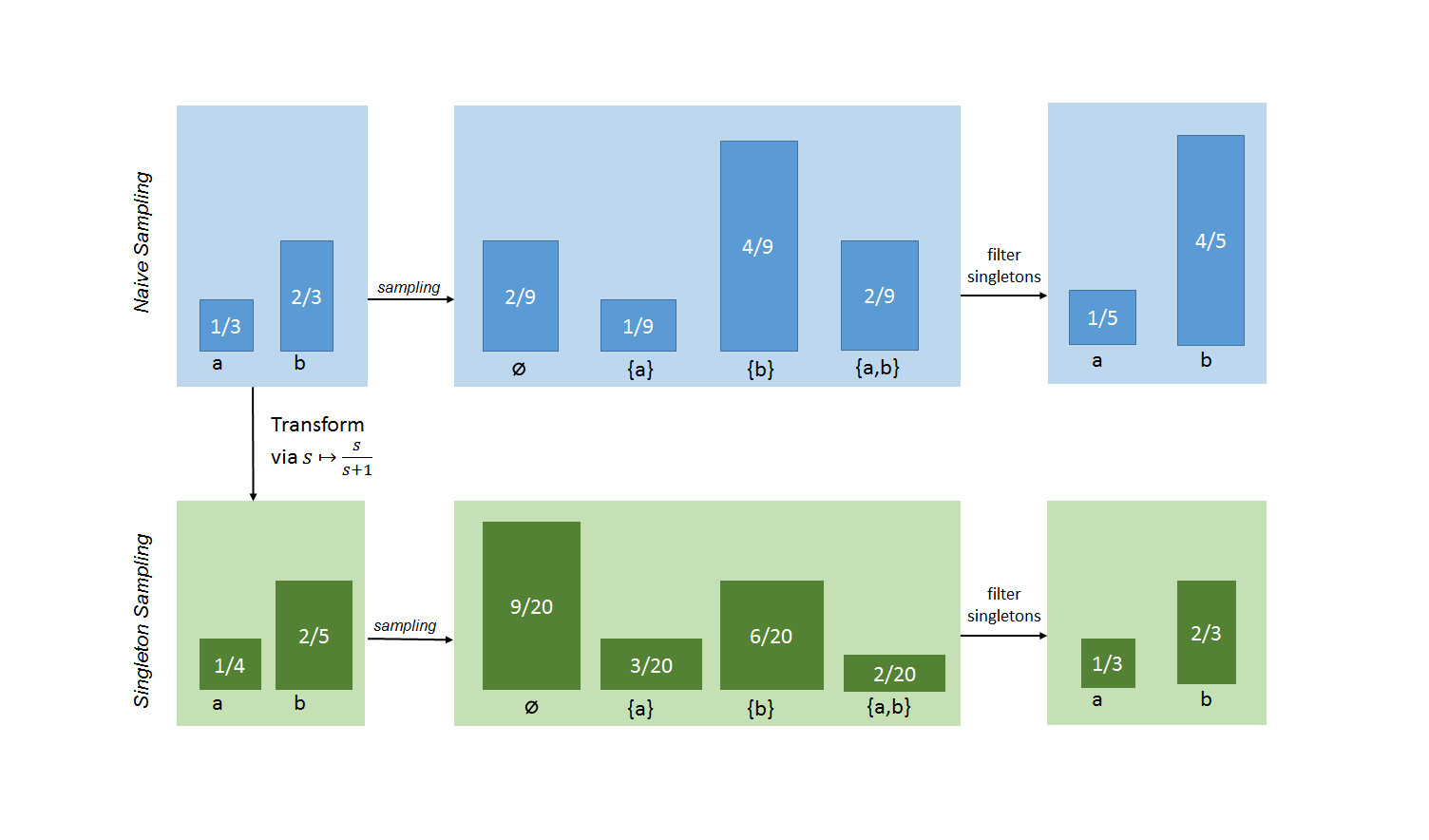}
  \caption{Singleton Sampling. Given a distribution over elements (column 1), we can independently sample element $n$ with probability $p_n$ to obtain a subset (column 2), which is then filtered to only obtain singletons (column 3).  However, the final distribution is distorted (row 1).  By pre-transforming the input distribution (row 2), we can recover the desired distribution through singleton sampling.}
  \label{fig:sampling}
\end{figure}

In Section~\ref{section:RAM} we restate a existing result to construct a coreset in the \texttt{RAM} model.
Section~\ref{section:computeSensitivity} will demonstrate how this construction can be accomplished while only looking at the $d$ rows of the $A^\top A$ (referencing the $n$ rows of $A$ in a very limited way).
Finally, in Section~\ref{section:streaming} we employ singleton sampling to maintain this coreset construction in the streaming setting.

\section{Preliminaries}

\textbf{Matrix Notation:}
For a matrix $M \in \mathbb{R}^{n \times d}$, let $M_{ij}$ denote the entry at row $i$ and column $j$.
$M_{\bullet j} \in \mathbb{R}^n$ denotes the $j^\text{th}$ column and $M_{j \bullet} \in \mathbb{R}^d$ denotes the $i^\text{th}$ row.  
To simplify notation throughout the proofs, we let $m_i = M_{i\bullet}^\top$ (and likewise: $a_i = A_{i\bullet}^\top$ and $u_i = U_{i\bullet}^\top$).
For a matrix $M \in \R^{n \times d}$ the term \textit{orthonormal} means that $MM^\top = I_n$ or $M^\top M = I_d$.
This generalizes the notion of an orthonormal square matrix, where both equalities are equivalent.
When $n \neq d$, at most one of the equalities can hold.

\begin{definition}[Squared-Distance] \label{def:distance}
For a vector $a' \in \R^d$ and a compact set $S \subset \R^d$, the squared-distance between $a'$ and $S$ is:
$$\mu(a',S) = \min_{s \in S} \|a'-s\|_2^2$$
where $\|\cdot\|_2$ is the $\ell_2$-norm.
This notation is overloaded for a matrix $A \in \R^{n \times d}$ as:
$$\mu(A,S) = \sum_{i=1}^n \mu(a_i, S)$$
\end{definition}

The above definition is the cost function used to define a coreset for SVD.  The input to the SVD problem is a matrix $A \in \R^{n \times d}$ which can be considered as a set of $n$ points in $\R^d$, namely $\{a_1, \ldots, a_n\}$.

\begin{definition}[Coreset for SVD]
Let $A \in \R^{n \times d}$ and $\epsilon \in [0,1]$.
A $\epsilon$-coreset for $A$ is matrix $B \in \R^{m \times d}$ such that for every affine $(d-1)$-subspace $S \subset \R^d$ we have:
$$ |\mu(A,S) - \mu(B,S)| \le \epsilon \mu(A,S)$$
The size of the coreset $B$ is $m$ (the number of rows in $B$).
When for each $j \in [m]$ we can write $b_j = w a_i$ for some value $w > 0$ and $i \in [n]$, we say that $B$ is a weighted subset of the input.
\end{definition}

If a coreset is a weighted subset of the input, then sparsity is preserved and the coreset admits a sparse representation whenever the input does.


\section{Coreset Construction}
\label{section:RAM}

In this section we show how to build a coreset for SVD, leaving computational concerns for the next section.
Let $\mathcal{S}_d$ be the set of all $d-1$ dimensional subspaces of $\R^d$.
\begin{definition}[Sensitivity] \label{def:sensitivity}
Let $A \in \R^{n \times d}$.
The sensitivity of a row $a_i$ with respect to $A$ is:
$$s_A(a_i) = \max_{S \in \mathcal{S}_d} \frac{\mu(a_i, S)}{\mu(A,S)}$$
\end{definition}

One can immediately observe that the sensitivity of any row lies in the interval $[0, 1]$.
We draw upon the result of~\cite{CP15}, who refer to the sensitivity by the name \textit{Lewis weight}.
The sensitivity has also been used by the term \textit{leverage score} in other publications.

\begin{theorem}[$\ell_2$ Matrix Concentration Bound from~\cite{CP15}] \label{thm:coreset}
Let $A \in \mathbb{R}^{n \times d}$.
Let $t = \sum_{i \in [n]} s_A(a_i)$.
Define a distribution over $A$ such that $a_i$ has weight $s_A(a_i) / t$.
Sample $m$ rows i.i.d. where:
$$ m \ge \frac{3t}{\epsilon^2} \left(\log_2 d + \ln \frac{1}{\delta} \right)$$
Construct a matrix $C \in \mathbb{R}^{m \times d}$ such each sampled row $a_i \in A$ corresponds to a row $\frac{t}{m s_A(a_i)} a_i \in C$.
With probability $1 - \delta$, $C$ is an $\epsilon$-coreset for $A$.
\end{theorem}

Using Theorem~\ref{thm:coreset}, the problem of constructing a coreset in the \texttt{RAM} model is reduced to computing $s_A(a_i)$ for each row $a_i \in A$.
The next section addresses this task.
Then in Section~\ref{section:streaming} we will export this construction to the streaming model.

\section{Computing Sensitivity}
\label{section:computeSensitivity}

We will show the construction in the \texttt{RAM} model, using only a restricted part of the input $A^\top A$, and then extend this to the streaming model in the next section.

\begin{theorem} \label{thm:ram}
Algorithm~\ref{alg:offline} takes $A^\top A \in \R^{d \times d}$ as input, terminates in $O(d^3)$ time, and outputs $Z = \mathcal{B}(A^\top A) \in \R^{d \times r}$ (where $r$ is the rank of $A$) such that $s(a_i) = \|A_{i \bullet} Z\|$ for each $i \in [n]$.
\end{theorem}


Before presenting Algorithm~\ref{alg:offline}, we show how one could compute the sensitivity by having all of $A$ (and not only $A^\top A$) in memory.

\begin{definition}[Thin SVD]
Let $A \in \R^{n \times d}$ be a matrix of rank $r$.  A Thin Singular Value Decomposition of $A$ is a (not necessarily unique) decomposition $A = U D V^\top$ for matrices $U \in \R^{n \times r}$, $D \in \R^{r \times r}$, and $V \in \R^{d \times r}$ such that:
\begin{itemize}
    \item $U^\top U = V^\top V = I_r$
    \item $D \in \R^{r \times r}$ is a diagonal matrix of positive values in non-increasing order
\end{itemize}
\end{definition}

It is well-known that the matrix $D$ is unique, so we write $D_A$ to specify the middle matrix in a Thin SVD of $A$.
Algorithms such as in~\cite{thinSVDTextbook} exist to compute the Thin SVD in $O(d^2)$ time (the time is in terms of number of operations on $d$-dimensional vectors).   
\begin{lemma}[Corollary of Lemma 3.1 of~\cite{maalouf2019tight}] \label{lem:sensIsNorm}
Let $UDV^\top$ be a Thin SVD of a matrix $A$.  The sensitivity (see Definition~\ref{def:sensitivity}) of row $a_i$ is $\|u_i\|^2$.
\end{lemma}

The matrices $U$ and $V$ are not unique, but Lemma~\ref{lem:sensIsNorm} implies that $\|u_i\|$ is invariant under any possible choice of $U$.
Therefore one can compute the sensitivity of a point by computing any Thin SVD and then taking the norm of each row of $U$.
We now turn our attention to proving Theorem~\ref{thm:ram}, where $\mathcal{B}(A^\top A)$ is not unique but the value obtained for the sensitivity is unique.

\begin{fact} \label{fact:rank}
The rank of $A^\top A$ is equal to the rank of $A$.
\end{fact}
\begin{proof}
Let $UDV^\top$ be a Thin SVD of $A$.
Then $A^\top A$ has a Thin SVD of $V D_A^2 V^\top$.
The result follows since all diagonal entries of $D_A^2$ are non-zero.
\end{proof}

We now walk through Algorithm~\ref{alg:offline} which computes the matrix $Z = \mathcal{B}(A^\top A)$.
Line~\ref{line:svd} computes a Thin SVD in  in $O(d^2)$ time.
$A^\top A$ is symmetric positive-definite, so it is guaranteed that $U = V$.
By uniqueness of singular values, we have that $\Lambda = D_{A^\top A} = D_A^2$.
On Line~\ref{line:root} we compute $D_A$ by taking the positive square-root of the entries of $D_A^2$, which are all real since $D_A^2$ has positive diagonal entries.
Let $r$ denote the rank of $A$.
$D_A V^\top \in \R^{r \times d}$ has rank $r$ and $r \le d$, so a right-inverse $Z \in \R^{d \times r}$ exists such that $D_A V^\top Z = I_r$ (see Line~\ref{line:rightInverse}).

We can therefore decompose $A = (AZ) D_A V^\top$.  If this is a Thin SVD, we can calculate the sensitivity of any row of $A$ by using $Z$.
First, we must verify that it is indeed a Thin SVD; it satisfies all properties but it remains to prove that the columns of $AZ$ are orthonormal.

\begin{lemma} \label{lem:orthCols}
For any matrix $Z$ output by Algorithm~\ref{alg:offline},
$AZ$ is orthonormal.
\end{lemma}
\begin{proof}
Algorithm~\ref{alg:offline} has exactly two sources of ambiguity: the choice of $V$ on Line~\ref{line:svd} and the choice of $Z$ on Line~\ref{line:rightInverse} (if $r < d$ the right-inverse is not unique).
The ambiguity of $Z$ is irrelevant since the value of $AZ$ is invariant under any choice.
It remains to prove the lemma under any choice of $V$ on Line~\ref{line:svd}.

Let $(U', D_A, V')$ be a Thin SVD of $A$.
Then $A^\top A = V' D_A^2 V'^\top$, and therefore the choice $V = V'$ is a possible outcome of Line~\ref{line:svd}.

For each diagonal entry $\lambda$ in $\Lambda$, the eigenspace $E_\lambda$ is unique.
Let us temporarily fix $\lambda$.
Moreover, if $\lambda$ occured in slots $j$ through $j'$ then $E_\lambda = \text{span}(V_{\bullet j}, \ldots, V_{\bullet j'})$.
Let $s = j' - j + 1 = \dim(E_\lambda)$.
Let $V_\lambda \in \R^{d \times s}$ denote the truncation of the $V$ to just columns $j$ through $j'$.
By uniqueness of the space $E_\lambda$, the matrix $V_\lambda$ is unique up to right-multiplication by an orthonormal matrix $X_\lambda \in \R^{s \times s}$.
In other words, $V'_\lambda = V_\lambda X_\lambda$.
Let $X$ be the block-diagonal $r \times r$ matrix of orthonormal transformations in each eigenspace from $V$ to $V'$, namely $V' = VX$.
Observe since each $X_\lambda$ multiplies with $\lambda I_s$ in $D$, we have the commuting relationship $XD = DX$.
Then we have $DV'^\top = D X^\top V^\top = X^\top DV^\top$.

Now consider the $U(V)$ and $U(V')$, the matrices $U$ we get from using $V$ or $V'$, respectively.
Observe the orthonormal relation $U(V) = U(V') X^\top$.
Since the columns of $U(V')$ and $X^\top$ are orthonormal, we conclude that $U(V)$ also has orthonormal columns.
\end{proof}

Lemma~\ref{lem:orthCols} ensures that $(AZ, D, V)$ is a Thin SVD of $A$ where the matrices $Z$, $D$, and $V$ are taken from any realization of Algorithm~\ref{alg:offline}.
We therefore conclude by Lemma~\ref{lem:sensIsNorm} that $s(a_i) = \|A_{i\bullet}Z\|$.
This completes the proof of Theorem~\ref{thm:ram}.

\begin{algorithm}
\caption{Input: matrix $\Psi \in \mathbb{R}^{d \times d}$}
\label{alg:offline}
\begin{algorithmic}[1]

\State $(U, \Lambda, V) \gets$ a Thin SVD of $\Psi$ \label{line:svd}

\State $D \gets \sqrt{\Lambda}$ \label{line:root}
 \Comment{$D \in \mathbb{R}^{r \times r}$ where $r = \texttt{rank}(\Psi) \le d$}

\State $Z \gets$ a right-inverse of $DV^\top$ \label{line:rightInverse} \Comment{$Z \in \mathbb{R}^{r \times d}$}

\State \Return $Z$

\end{algorithmic}
\end{algorithm}

In conclusion, we denote deterministic Algorithm~\ref{alg:offline} to take $A^\top A$ and output a matrix $Z$.
Although $Z$ may not be unique, it has the required invariant of $s(a_i) = \|A_{i \bullet} Z\|^2$ for each $i \in [n]$.

\section{Streaming Algorithm}
\label{section:streaming}

\begin{definition}[Streaming $\epsilon$-coreset]
A streaming $\epsilon$-coreset is an algorithm that receives a stream of elements.
After receiving each element, it returns an $\epsilon$-coreset for the prefix of elements received so far.
\end{definition}

Now the rows of $A$ will arrive in a stream.
Let $a_i$ denote the $i^\text{th}$ row of $A$, and
let $A_n$ denote the matrix $A$ after the first $n$ rows have arrived.
On Line~\ref{line:sensitivity}, $s_n(a_i)$ is the sensitivity of row $a_i$ with respect to $A_n$ (rows $a_1, \ldots, a_n$).


\begin{algorithm}
\caption{Input: $\epsilon \in (0,1)$, $\delta \in (0,1)$, stream of points in $\R^d$}
 \label{alg:main}
\begin{algorithmic}[1]

\State $m \gets \lceil 3 d \epsilon^{-2} (\log_2^2 d + \ln(2/\delta)) \rceil$
\label{line:defM}

\State $Y \gets \{1, \ldots, 8m \}$
\Comment{index over the samplers}

\For {each $y \in Y$}
    \State $M_{y} \gets \emptyset$ \Comment{each $M_y$ stores a sample of rows}
\EndFor \label{line:initEnd}

	\State $\Psi_0 \gets 0_{d \times d}$ \Comment{$\Psi_i = A_i^\top A_i$ for all $i$}

\For {each row $a_n \in \R^d$ from the stream}
\label{line:mainStart}

	\State $\Psi_n \gets \Psi_{n-1} + a_n a_n^\top$ \label{line:updateATA}
	
	\For {each $y \in Y$}
	    \State $u_{y}(a_n) \gets$ uniform random number in $[0,1)$

	    \State $M_{y} \gets M_{y} \cup \{a_n\}$
    \EndFor

    \State $Z_n \gets \mathcal{B}(\Psi_n)$ \Comment{Algorithm~\ref{alg:offline}}
    
    \State $r_n \gets $ number of columns of $Z_n$ \label{line:rank}
    
    \For{each $a_i \in \cup_{y} M_{y}$}

        \State $s_{n}(a_i) \gets \|Z_n ^ \top a_i\|^2$ \label{line:sensitivity} \Comment{compute sensitivity}
    \EndFor
    
    \For {each $y \in Y$}
	    \For {each $a_i \in M_{y}$}
	        \If {$u_y(a_i) >  \frac{s_n(a_i)}{s_n(a_i) + r_n}$} \label{line:delete}
	            \State Delete $a_i$ from $M_{y}$
	       \EndIf
	   \EndFor
	\EndFor \label{line:endDeletions}
	
	\State $\Gamma_n \gets \{y \in Y : |M_y| = 1\}$
	\label{line:setG} \Comment{index over singleton samples}

    \State $Q_n \gets \cup_{y \in \Gamma_n} M_{y}$ \label{line:buildSet}
    \Comment{the union of all samples containing exactly one row}
    
    \For{each $a_i \in Q_n$} 
	    \State $w_n(a_i) \gets \frac{r_n}{|Q_n|  s_n(a_i)}$ \label{line:weight}
	\EndFor
	
	\State \Return $(Q_n, w_n)$
	\Comment{the $\epsilon$-coreset for $A_n$}
\EndFor \label{line:mainEnd}
\end{algorithmic}
\end{algorithm}


\begin{figure}
  \includegraphics[width=\linewidth]{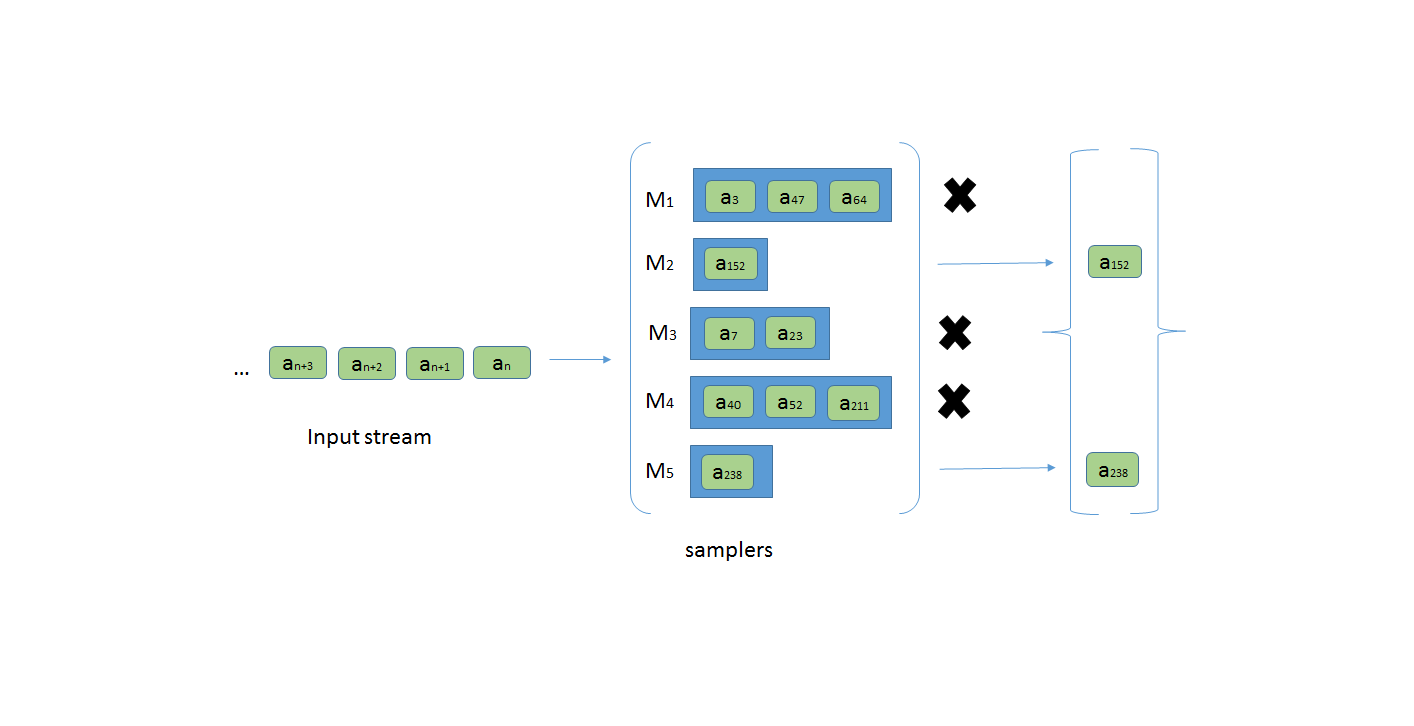}
  \caption{The flow for updating the samplers after each update in Algorithm~\ref{alg:main}.  The new point is first added to each sampler, and then the deletion condition is checked for every point in every sampler.}
  \label{fig:draw2}
\end{figure}

\begin{figure}
  \includegraphics[width=\linewidth]{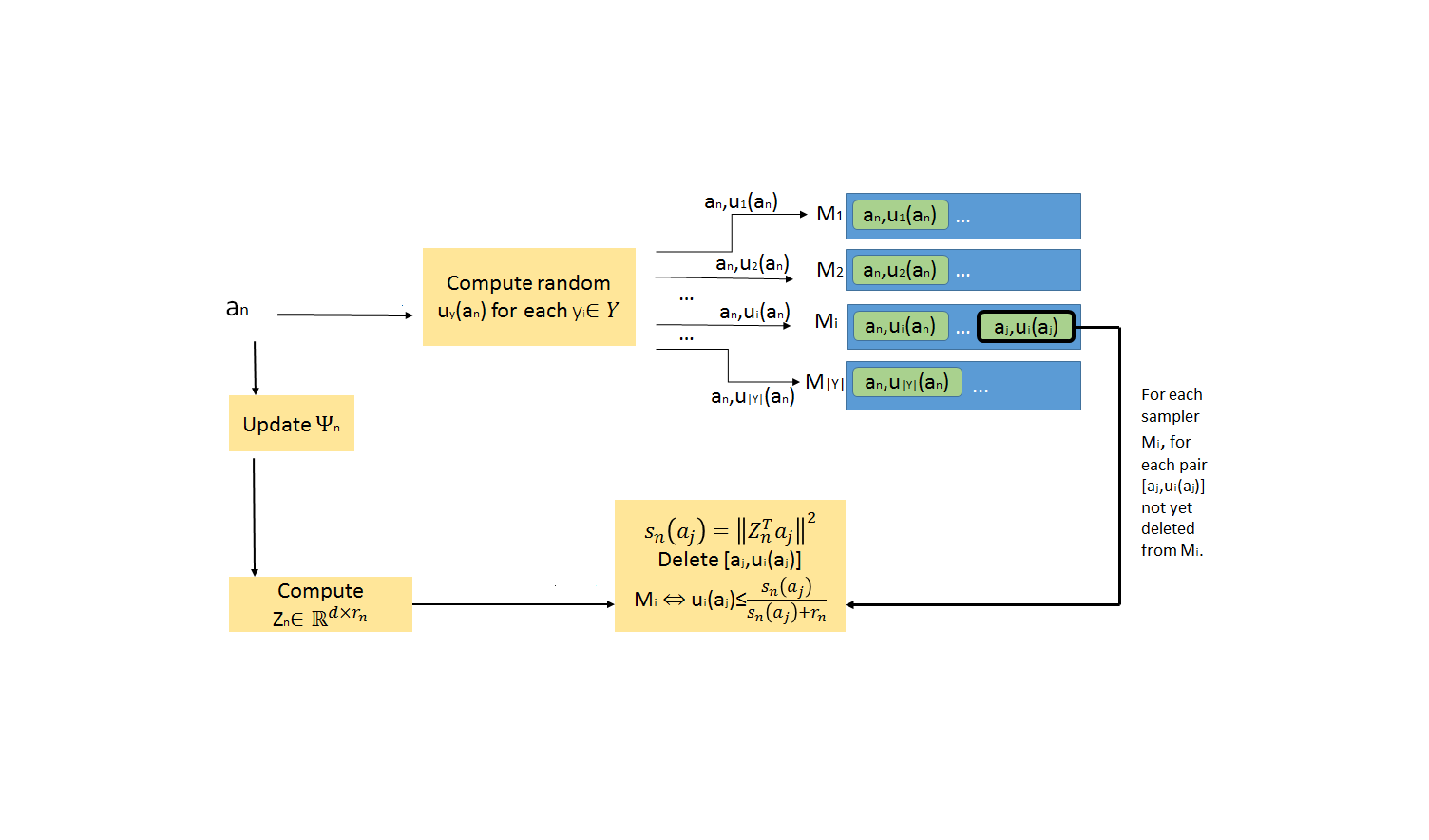}
  \caption{The input stream is read into an array of $8m$ samplers $\{M_y\}_{y \in Y}$.  At each step, the samplers containing exactly one element are combined to obtain a coreset.}
  \label{fig:draw3}
\end{figure}

\begin{theorem}[Main Theorem] \label{thm:streaming}
Let $\epsilon, \delta \in (0,1)$ and let $A$ be a stream of rows in $\R^d$.
After receiving each row of $A$, Algorithm~\ref{alg:main} returns an $\epsilon$-coreset for $A$ with probability $1 - \delta$.
The space is $O(\frac{d}{\epsilon^2} (\log^2 d + \log \frac{1}{\delta}))$ rows and the worst-case update time is $O(d^2)$.
\end{theorem}

We denote the rank of $A_n$ by $r_n$ (computed on Line~\ref{line:rank}).
The next lemma shows that we know the total sensitivity
exactly.

\begin{lemma} \label{lem:totalSens}
The total sensitivity of $A_n$ is $r_n$.
\end{lemma}
\begin{proof}
We must prove $\sum_{i = 1}^n s_n(a_i) = r_n$.
For any Thin SVD $A_n = UDV^T$, we have that $s_n(a_i) = \|u_i\|^2$ by Lemma~\ref{lem:sensIsNorm}, so it suffices to compute $\sum_{i = 1}^n \|u_i\|^2$.
Since the columns of $U$ are orthonormal, we argue as follows:
\begin{align*}
\sum_{i = 1}^n \|u_i\|^2 &= \sum_{i = 1}^n \sum_{j = 1}^{r_n} U_{ij}^2 \\
&= \sum_{j = 1}^{r_n} \sum_{i = 1}^n U_{ij}^2 \\
&= \sum_{j = 1}^{r_n} \|u_j\|^2 \\
&= \sum_{j = 1}^{r_n} 1 \\
&= {r_n}
\end{align*}
\end{proof}

Note that $r_n$ (the rank of $A_n$) equals the dimension of the row-space of $A_n$ and therefore cannot decrease as the stream progresses.
Moreover, $r_n$ is trivially upper-bounded by $d$, as each row is embedded in $\mathbb{R}^d$.

Algorithm~\ref{alg:main} keeps a set of samples $\{M_y\}$, where each $M_y$ stores a sample of rows of $A$.
We refer to any sample $M_y$ that stores exactly one row as a \textit{singleton}.
As seen by Lines~\ref{line:setG}-\ref{line:buildSet}, any non-singletons are not used to build the coreset $(Q_n, w_n)$.
We first prove that every singleton is drawn from the distribution required by Theorem~\ref{thm:coreset}, namely that $a_i$ is sampled with probability $s_n(a_i) / r_n$.
In what follows we let $M_y^{(n)}$ denote the state of $M_y$ after row $a_n$ is processed.

\begin{lemma}
For any $i \in [n]$, $Pr(M_y^{(n)} = \{ a_i \} : |M_y^{(n)}| = 1) = s_n(a_i) / r_n$
\end{lemma}
\begin{proof}
Let $\gamma_\ell = Pr(a_i \in M_y^{(n)}) = \frac{s_n(a_\ell)}{s_n(a_\ell) + r_n}$ and define $\xi = \prod_{\ell=1}^n ( 1 - \gamma_\ell)$.
Observe that $Pr(M^{(n)}_{y} = \{a_i\}) = \xi \frac{\gamma_i}{1 - \gamma_i} = \xi s_n(a_i)/r_n$.
To condition upon the event that $M_y^{(n)}$ contains only a single element, we divide this probability by $Pr(|M_y^{(n)}| = 1) = \sum_{\ell = 1}^n Pr(M^{(n)}_{y} = \{p_\ell\})$ to obtain $s_n(a_i) / \sum_{\ell = 1}^n s_n(a_\ell)$.
The result then follows from Lemma~\ref{lem:totalSens}.
\end{proof}

Each sample $M_y$ should have a size concentrated around $1$ with high probability.  If the size is too far away from $1$, there will not be enough singletons to build a coreset.  Also, if the expected size is too large, the algorithm will require too much space.

\begin{lemma} \label{lem:expected}
The expected value of $|M_y^{(n)}|$ is $1$.
\end{lemma}
\begin{proof}
As in Lemma~\ref{lem:singleton}, define $\gamma_\ell = \frac{s_n(a_\ell)}{s_n(a_\ell) + r_n}$.
Observe that $s_n(a_\ell) \ge 0$ implies $\gamma_\ell \le s_n(a_\ell) / r_n$.
It follows from Line~\ref{line:delete} that $Pr(a_\ell \in M_y) = \gamma_\ell \le s_n(a_\ell) / r_n$.
The expected value of $|M_y|$ is therefore at most $\sum_{\ell =1}^n \frac{1}{r_n} s_n(a_\ell) = 1$.
\end{proof}

We now show that any $M_y$ will be a singleton with probability at least $\frac{1}{4}$.  This is essential for showing that, in aggregate, there will be at least $m$ singletons (Lemma~\ref{lem:atLeastM}) so that we can build a coreset.

\begin{lemma}
\label{lem:singleton}
For any $y \in Y$, $P(|M^{(n)}_y| = 1) \ge \frac{1}{4}$.
\end{lemma}
\begin{proof}
Markov's inequality to Lemma~\ref{lem:expected} yields $Pr(|M_y| \ge 2) \le \frac{1}{2}$.

\begin{align*}
P(|M_y|=1) &= \sum_{\ell = 1}^n P(M_y = \{a_\ell\}) \\
&= \sum_{\ell = 1}^n \gamma_\ell \prod_{z \neq \ell} (1 -  \gamma_z) \\
&= \sum_{\ell = 1}^n \frac{\gamma_\ell}{1 - \gamma_\ell} \prod_{z=1}^n (1 - \gamma_z) \\
&= \frac{1}{r_n} \left( \sum_{\ell = 1}^n  s_n(a_\ell) \right) \left(\prod_{z=1}^n (1 - \gamma_z) \right) \\
&= Pr(|M_y| = 0)
\end{align*}

Decomposing $1 = Pr(|M_y| = 0) + Pr(|M_y|=1) + Pr(|M_y| \ge 2)$, we substitute to obtain $1 \le 2 Pr(|M_y| = 1) + \frac{1}{2}$.  The result follows.
\end{proof}

Theorem~\ref{thm:coreset} guarantees that a coreset can be constructed provided that we have at least $m$ singletons.
As seen on Line~\ref{line:setG}, $|\Gamma_n|$ is the number of singletons after processing row $a_n$.

\begin{lemma}
\label{lem:atLeastM}
$|\Gamma_n| \ge m$ with probability at least $1 -\delta / 2$.
\end{lemma}
\begin{proof}
We see from Line~\ref{line:setG} that $|\Gamma_n|$ is a sum of $|Y|$ independent Bernoulli trials, each which succeeds with probability at least $\frac{1}{4}$ by Lemma~\ref{lem:singleton}.
By a Chernoff bound, $Pr(|\Gamma_n| \le \frac{1}{2} |Y| \frac{1}{4}) \le e^{-|Y| / 32}$.
We note that $|Y| / 32 = m / 4 \ge \ln(2 / \delta)$.  Plugging this into the Chernoff bound yields the result.
\end{proof}

In terms of space consumption, Algorithm~\ref{alg:main} stores $\Psi_n$ ($d$ rows), $Z_n$ ($r_n \le d$ rows), and the $\{M^{(n)}_y\}_{y \in Y}$.  These samples, of random size, are the dominating factor.
The next lemma bounds their aggregate size with high probability.

\begin{lemma} \label{lem:space}
Algorithm~\ref{alg:main} stores $O(\frac{d}{\epsilon^2} (\log^2 d + \log \frac{1}{\delta}))$ rows with probability $1 - \delta / 2$.
\end{lemma}
\begin{proof}
Lemma~\ref{lem:expected} implies that $\sum_{y \in Y} |M_y^{(n)}| = |Y| = 8m$.
Let the random variable $X_n$ denote the space (in rows) used by all $\{M^{(n)}_y\}_{y \in Y}$ after row $a_n$ arrives.
$X_n = \sum_{y \in Y} |M^{(n)}_{y}|$ is a sum of $|Y|n$ independent Bernoulli trials, each event being $a_i \in M^{(n)}_y$ for some $1 \le i \le n$ and $y \in Y$.
We have the Chernoff bound that $Pr(X_n \ge (1+\eta) \mu) \le e^{-\eta^2 \mu / (2 + \eta)}$ for any $\eta \ge 1$ where $\mu = E[X_n]$.
Using $\eta = 1$, this yields that $X_n < 16m$ with probability at least $1 - e^{-m/3} < \delta / 2$.
Substitute the value for $m$ from Line~\ref{line:defM} to finish the proof.
\end{proof}

\begin{proof}[Proof of Theorem~\ref{thm:streaming}]
For correctness, we will apply Theorem~\ref{thm:coreset} to the output $(Q_n, w_n)$.
The total sensitivity $t = r_n \le d$ by Lemma~\ref{lem:totalSens}.
The weighting is correct (Line~\ref{line:weight}).
It remains to show that $|Q_n| 
\ge \frac{cd}{\epsilon^2} \left( \log_2^2 d + \ln \frac{1}{\delta} \right)$.
Since $|Q_n| = |\Gamma_n|$, the result follows with probability $1 - \delta / 2$ by combining Lemma~\ref{lem:atLeastM} and Line~\ref{line:defM}.

As for update time, the dominating factor is computing the Thin SVD on Line~\ref{line:svd} which can be done in $O(d^2)$ time.  Then Lemma~\ref{lem:space} completes the proof of the theorem for subspaces.
The reduction from affine subspaces to proper subspaces can be found in Section 4 of~\cite{maalouf2019tight}.
\end{proof}

\section{Experiments}

\newcommand{\sca}{0.33}
\newcommand{\scal}{0.33}
\begin{figure}[!ht]
	\begin{subfigure}[h]{\scal\textwidth}
		\centering
		\includegraphics[scale=\scal]{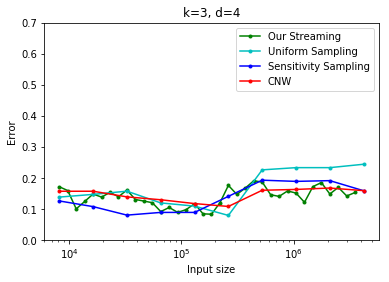}
		\caption{\label{wt1} $k=3, d=4$, error}
	\end{subfigure}
		\begin{subfigure}[h]{\scal\textwidth}
		\centering
		\includegraphics[scale=\scal]{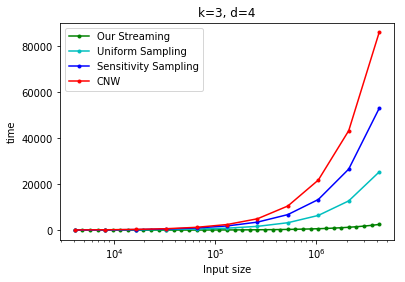}
		\caption{\label{wt2} $k=3, d=4$, time}
	\end{subfigure}
	\begin{subfigure}[h]{\scal\textwidth}
		\centering
		\includegraphics[scale=\scal]{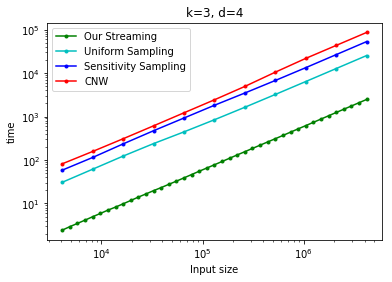}
		\caption{\label{wt3} $k=3, d=4$, time log}
	\end{subfigure}
	\begin{subfigure}[h]{\scal\textwidth}
		\centering
		\includegraphics[scale=\scal]{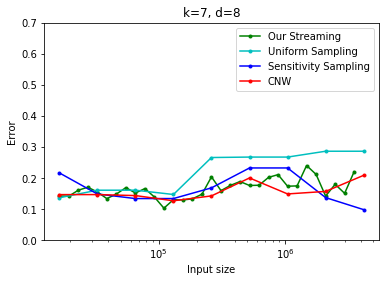}
		\caption{\label{wt4} $k=7, d=8$, error}
	\end{subfigure}
		\begin{subfigure}[h]{\scal\textwidth}
		\centering
		\includegraphics[scale=\scal]{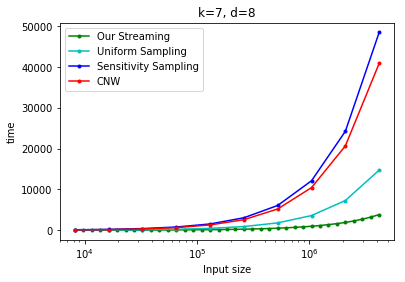}
		\caption{\label{wt5} $k=7, d=8$, time}
	\end{subfigure}
	\begin{subfigure}[h]{\scal\textwidth}
		\centering
		\includegraphics[scale=\scal]{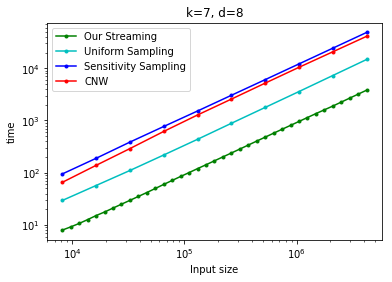}
		\caption{\label{wt6} $k=7, d=8$, time log}
	\end{subfigure}
		\begin{subfigure}[h]{\scal\textwidth}
		\centering
		\includegraphics[scale=\scal]{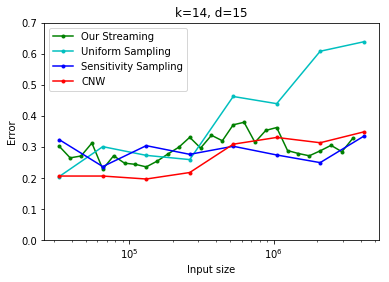}
		\caption{\label{wt7} $k=14, d=15$, error}
	\end{subfigure}
		\begin{subfigure}[h]{\scal\textwidth}
		\centering
		\includegraphics[scale=\scal]{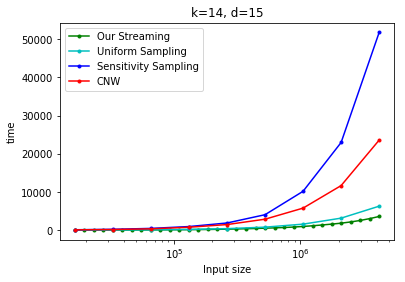}
		\caption{\label{wt8} $k=14, d=15$, time}
	\end{subfigure}
	\begin{subfigure}[h]{\scal\textwidth}
		\centering
		\includegraphics[scale=\scal]{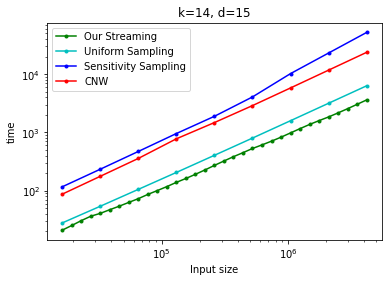}
		\caption{\label{wt9} $k=14, d=15$, time log}
	\end{subfigure}
	\begin{subfigure}[h]{\scal\textwidth}
		\centering
		\includegraphics[scale=\scal]{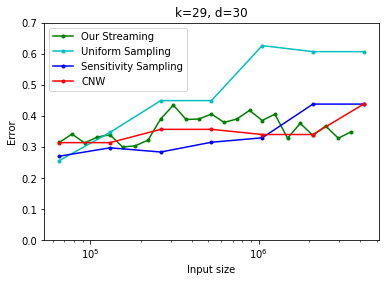}
		\caption{\label{wt10} $k=29, d=30$, error}
	\end{subfigure}
		\begin{subfigure}[h]{\scal\textwidth}
		\centering
		\includegraphics[scale=\scal]{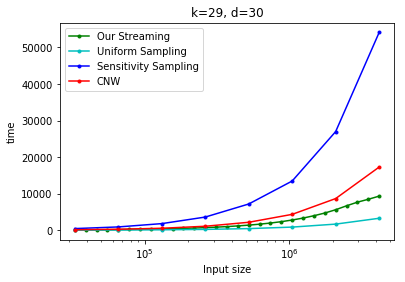}
		\caption{\label{wt11} $k=29, d=30$, time}
	\end{subfigure}
	\begin{subfigure}[h]{\scal\textwidth}
		\centering
		\includegraphics[scale=\scal]{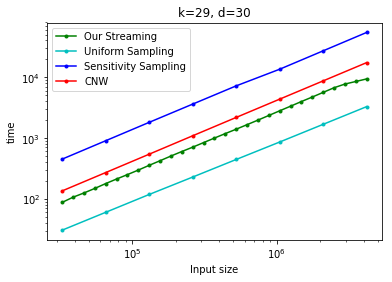}
		\caption{\label{wt12} $k=29, d=30$, time log}
	\end{subfigure}

\caption{\small\label{F4}Error and time comparison between two methods of streaming on full JL transform of Wikipedia, for different values of $k$ and $d$.}
\end{figure}

We then run experimental results that we summarize in this section.

\paragraph{Wikipedia Dataset}
We created a document-term matrix of Wikipedia (parsed enwiki-latest-pages -articles.xml.bz2-rss.xml from \cite{wic2019}), i.e. sparse matrix with 4624611 rows and 100k columns where each cell $(i,j)$ equals the value of how many appearances the word number $j$ has in article number $i$. We use a standard dictionary of the 100k most common words in Wikipedia \cite{dic2012}.

In order to compact the data into a small $d$, one applied on it a Johnson-Lindenstrauss (JL; see \cite{johnson1984extensions}) transform: We multiplied this each chunk from the BOW matrix by a randomized matrix of $100K$ rows and $d$ columns, and got a dense matrix of $n$ rows and $d$ columns.

\paragraph{Tree system}
We implemented a tree system that separates the $n$ points of the data into chunks of a desired size of coreset, called $m$. It uses consecutive chunks of the data, merge each pair of them, and uses a desired algorithm in order to reduce their dimensionality to a half. The process is described well in \cite{feldman2010coresets}. The result is a top coreset of the whole data, in a size of $d$.
We built such a system. Each streaming chunk is in the size of $d$, thus we had $log(n/d)=22$ floors.

\paragraph{Algorithms.}
We ran JL-Wiki two streaming methods: The one described at "Tree system", and the new method, implementation of Algorithm 2, labeled as "New Streaming". We implemented Algorithm 1 and our streaming algorithm; See Algorithm 2. 
The algorithms run on the data blocks of the streaming tree are three:
\begin{enumerate}
    \item A determinstic algorithm, appears in the proof of Theorem 5 of \cite{cohen2015optimal}, that guarantees a coreset with an error of $\epsilon$ in a size of $O(\frac{k}{\epsilon^2})$ (Labeled as "CNW"). 
    \item A non-uniform sampling algorithm, Algorithm 1 of \cite{maalouf2019tight}, that guarantees a tight bound of sensitivity for any $k$ and exact sensitivity  for $d=k-1$. Sampling points according to their sensitivities guarantees a coreset with an error of $\epsilon$ in a size of $O(\frac{d\log_2{d}}{\epsilon^2}).$ (Labeled as "Sensitivity Sampling").
    \item A uniform sampling method (Labeled as "Uniform Sampling").
\end{enumerate}
All algorithms were implemented in Python 3.6 via the Numpy library.
The "New Streaming" results were averaged over 10 experiments.
We ran four experiments with the following parameterization: 
\begin{itemize}
\item k=3, d=4
\item k=7, d=8
\item k=14, d=15
\item k=29, d=30
\end{itemize}

\paragraph{Hardware.}
A desktop, with an Intel i7-6850K CPU @ 3.60GHZ 64GB RAM. 
\paragraph{Results. }
We compared the error received for the different algorithms. We show the results in Figure \ref{F4} in $x$-logarithmic scale since the floors' sizes of the old streaming method differ multiplicatively.
In the old streaming, for every floor, we concatenated the leaves of the floor and measured the error between this subset to the original data. The error we determined was calculated by the formula $\frac{\norm{A-AV_C^TV_C}^2-\norm{A-AV_A^TV_A}^2}{\norm{A-AV_A^TV_A}^2}$, where $A$ is the recieved data matrix, $V_A$ received by SVD on A, and $V_C$ received by SVD on the top leaf recieved from steaming $A$; See \ref{wt1},\ref{wt4},\ref{wt7},\ref{wt10}. We also measured running times of each input size. The are showed in linear scale on \ref{wt2},\ref{wt5},\ref{wt8},\ref{wt11}, and in logarithmic scale on \ref{wt3},\ref{wt6},\ref{wt9},\ref{wt12}

\paragraph{Discussion. }
 One can notice in Figures \ref{F4} that the error is relatively similar for the new streaming as the streaming tree methods and as expected lower than the error of the Uniform Sampling.  In addition, running times of the new method are shorter in a scale or two than the streaming tree's blocks methods, and usually even shorter than those of Uniform Sampling. 

\bibliography{main.bbl}

@article{MSF19,
  author    = {Alaa Maalouf and
               Adiel Statman and
               Dan Feldman},
  title     = {Tight Sensitivity Bounds For Smaller Coresets},
  journal   = {CoRR},
  volume    = {abs/1907.01433},
  year      = {2019},
  url       = {http://arxiv.org/abs/1907.01433},
  archivePrefix = {arXiv},
  eprint    = {1907.01433},
  timestamp = {Mon, 08 Jul 2019 14:12:33 +0200},
  biburl    = {https://dblp.org/rec/bib/journals/corr/abs-1907-01433},
  bibsource = {dblp computer science bibliography, https://dblp.org}
}
@article{maalouf2019tight,
  title={Tight Sensitivity Bounds For Smaller Coresets},
  author={Maalouf, Alaa and Statman, Adiel and Feldman, Dan},
  journal={arXiv preprint arXiv:1907.01433},
  year={2019}
}

@article{Woodruff2012,
author = {Drineas, Petros and Magdon-Ismail, Malik and Mahoney, Michael W. and Woodruff, David P.},
title = {Fast Approximation of Matrix Coherence and Statistical Leverage},
year = {2012},
issue_date = {January 2012},
publisher = {JMLR.org},
volume = {13},
number = {1},
issn = {1532-4435},
journal = {J. Mach. Learn. Res.},
month = dec,
pages = {3475–3506},
numpages = {32},
keywords = {randomized algorithm, statistical leverage, matrix coherence}
}
  


@article{cohen2015optimal,
  title={Optimal approximate matrix product in terms of stable rank},
  author={Cohen, Michael B and Nelson, Jelani and Woodruff, David P},
  journal={arXiv preprint arXiv:1507.02268},
  year={2015}
}

@MISC{dic2012,
author={Dan Church},
month={March},
year = {2012},
howpublished={\url{https://gist.github.com/h3xx/1976236}}
}

@MISC{wic2019,
author={Wikimedia},
month={June},
year = {2019},
howpublished={\url{https://dumps.wikimedia.org/enwiki/latest/}}
}

@article{johnson1984extensions,
  title={Extensions of Lipschitz mappings into a Hilbert space},
  author={Johnson, William B and Lindenstrauss, Joram},
  journal={Contemporary mathematics},
  volume={26},
  number={189-206},
  pages={1},
  year={1984}
}

@inproceedings{FL11,
author = {Feldman, Dan and Langberg, Michael},
title = {A Unified Framework for Approximating and Clustering Data},
year = {2011},
isbn = {9781450306911},
publisher = {Association for Computing Machinery},
address = {New York, NY, USA},
url = {https://doi.org/10.1145/1993636.1993712},
doi = {10.1145/1993636.1993712},
booktitle = {Proceedings of the Forty-Third Annual ACM Symposium on Theory of Computing},
pages = {569–578},
numpages = {10},
keywords = {clustering, regression, cur, epsilon-approximation, svd, k-means, pac-learning, approximating, epsilon-nets, k-median, pca, coresets},
location = {San Jose, California, USA},
series = {STOC ’11}
}
  


@inproceedings{FSS13,
author = {Feldman, Dan and Schmidt, Melanie and Sohler, Christian},
title = {Turning Big Data into Tiny Data: Constant-Size Coresets for k-Means, PCA and Projective Clustering},
year = {2013},
isbn = {9781611972511},
publisher = {Society for Industrial and Applied Mathematics},
address = {USA},
booktitle = {Proceedings of the Twenty-Fourth Annual ACM-SIAM Symposium on Discrete Algorithms},
pages = {1434–1453},
numpages = {20},
location = {New Orleans, Louisiana},
series = {SODA ’13}
}
  
@inproceedings{CP15,
author = {Cohen, Michael B. and Peng, Richard},
title = {Lp Row Sampling by Lewis Weights},
year = {2015},
isbn = {9781450335362},
publisher = {Association for Computing Machinery},
address = {New York, NY, USA},
booktitle = {Proceedings of the Forty-Seventh Annual ACM Symposium on Theory of Computing},
pages = {183–192},
numpages = {10},
keywords = {lewis weights, matrix concentration bounds, lp regression, row sampling},
location = {Portland, Oregon, USA},
series = {STOC ’15}
}
  

@inproceedings{feldman2010coresets,
  title={Coresets and sketches for high dimensional subspace approximation problems},
  author={Feldman, Dan and Monemizadeh, Morteza and Sohler, Christian and Woodruff, David P},
  booktitle={Proceedings of the twenty-first annual ACM-SIAM symposium on Discrete Algorithms},
  pages={630--649},
  year={2010},
  organization={Society for Industrial and Applied Mathematics}
}

@inproceedings{FVR16,
  author    = {Dan Feldman and
               Mikhail Volkov and
               Daniela Rus},
  title     = {Dimensionality Reduction of Massive Sparse Datasets Using Coresets},
  booktitle = {Advances in Neural Information Processing Systems 29: Annual Conference
               on Neural Information Processing Systems 2016, December 5-10, 2016,
               Barcelona, Spain},
  pages     = {2766--2774},
  year      = {2016},
}

@article{CE14,
author = {Cohen, Michael and Elder, Sam and Musco, Cameron and Musco, Christopher and Persu, Madalina},
title = {Dimensionality Reduction for K-Means Clustering and Low Rank Approximation},
year = {2015},
isbn = {9781450335362},
publisher = {Association for Computing Machinery},
address = {New York, NY, USA},
booktitle = {Proceedings of the Forty-Seventh Annual ACM Symposium on Theory of Computing},
pages = {163–172},
numpages = {10},
location = {Portland, Oregon, USA},
journal = {STOC ’15}
}

@inproceedings{BFL19,
  author    = {Vladimir Braverman and
               Dan Feldman and
               Harry Lang and
               Daniela Rus},
  title     = {Streaming Coreset Constructions for M-Estimators},
  booktitle = {Approximation, Randomization, and Combinatorial Optimization. Algorithms
               and Techniques, {APPROX/RANDOM} 2019, September 20-22, 2019, Massachusetts
               Institute of Technology, Cambridge, MA, {USA}},
  pages     = {62:1--62:15},
  year      = {2019},
}

@BOOK{thinSVDTextbook,
  TITLE = {Numerical Linear Algebra},
  AUTHOR = {Lloyd Trefethen and David Bau III},
  YEAR = {1997},
  PUBLISHER = {SIAM},
}


\begin{thebibliography}{10}

\bibitem{BFL19}
Vladimir Braverman, Dan Feldman, Harry Lang, and Daniela Rus.
\newblock Streaming coreset constructions for m-estimators.
\newblock In {\em Approximation, Randomization, and Combinatorial Optimization.
  Algorithms and Techniques, {APPROX/RANDOM} 2019, September 20-22, 2019,
  Massachusetts Institute of Technology, Cambridge, MA, {USA}}, pages
  62:1--62:15, 2019.

\bibitem{dic2012}
Dan Church.
\newblock \url{https://gist.github.com/h3xx/1976236}, March 2012.

\bibitem{CE14}
Michael Cohen, Sam Elder, Cameron Musco, Christopher Musco, and Madalina Persu.
\newblock Dimensionality reduction for k-means clustering and low rank
  approximation.
\newblock {\em STOC ’15}, page 163–172, 2015.

\bibitem{cohen2015optimal}
Michael~B Cohen, Jelani Nelson, and David~P Woodruff.
\newblock Optimal approximate matrix product in terms of stable rank.
\newblock {\em arXiv preprint arXiv:1507.02268}, 2015.

\bibitem{CP15}
Michael~B. Cohen and Richard Peng.
\newblock Lp row sampling by lewis weights.
\newblock In {\em Proceedings of the Forty-Seventh Annual ACM Symposium on
  Theory of Computing}, STOC ’15, page 183–192, New York, NY, USA, 2015.
  Association for Computing Machinery.

\bibitem{Woodruff2012}
Petros Drineas, Malik Magdon-Ismail, Michael~W. Mahoney, and David~P. Woodruff.
\newblock Fast approximation of matrix coherence and statistical leverage.
\newblock {\em J. Mach. Learn. Res.}, 13(1):3475–3506, December 2012.

\bibitem{FL11}
Dan Feldman and Michael Langberg.
\newblock A unified framework for approximating and clustering data.
\newblock In {\em Proceedings of the Forty-Third Annual ACM Symposium on Theory
  of Computing}, STOC ’11, page 569–578, New York, NY, USA, 2011.
  Association for Computing Machinery.

\bibitem{feldman2010coresets}
Dan Feldman, Morteza Monemizadeh, Christian Sohler, and David~P Woodruff.
\newblock Coresets and sketches for high dimensional subspace approximation
  problems.
\newblock In {\em Proceedings of the twenty-first annual ACM-SIAM symposium on
  Discrete Algorithms}, pages 630--649. Society for Industrial and Applied
  Mathematics, 2010.

\bibitem{FSS13}
Dan Feldman, Melanie Schmidt, and Christian Sohler.
\newblock Turning big data into tiny data: Constant-size coresets for k-means,
  pca and projective clustering.
\newblock In {\em Proceedings of the Twenty-Fourth Annual ACM-SIAM Symposium on
  Discrete Algorithms}, SODA ’13, page 1434–1453, USA, 2013. Society for
  Industrial and Applied Mathematics.

\bibitem{FVR16}
Dan Feldman, Mikhail Volkov, and Daniela Rus.
\newblock Dimensionality reduction of massive sparse datasets using coresets.
\newblock In {\em Advances in Neural Information Processing Systems 29: Annual
  Conference on Neural Information Processing Systems 2016, December 5-10,
  2016, Barcelona, Spain}, pages 2766--2774, 2016.

\bibitem{johnson1984extensions}
William~B Johnson and Joram Lindenstrauss.
\newblock Extensions of lipschitz mappings into a hilbert space.
\newblock {\em Contemporary mathematics}, 26(189-206):1, 1984.

\bibitem{maalouf2019tight}
Alaa Maalouf, Adiel Statman, and Dan Feldman.
\newblock Tight sensitivity bounds for smaller coresets.
\newblock {\em arXiv preprint arXiv:1907.01433}, 2019.

\bibitem{thinSVDTextbook}
Lloyd Trefethen and David~Bau III.
\newblock {\em Numerical Linear Algebra}.
\newblock SIAM, 1997.

\bibitem{wic2019}
Wikimedia.
\newblock \url{https://dumps.wikimedia.org/enwiki/latest/}, June 2019.

\end{thebibliography}
\bibliographystyle{plain}

\end{document}